\documentclass[12pt]{article}

\usepackage{tqian-latex-style}

\usepackage[toc,page]{appendix}

\usepackage{thmtools}

\declaretheorem[
  name=Theorem,
  Refname={Theorem,Theorems},
  numberwithin=section]{thm}
\declaretheorem[
  name=Lemma,
  Refname={Lemma,Lemmas},
  sibling=thm]{lem}

\declaretheorem[
  name=Assumption,
  Refname={Assumption,Assumptions}]{asu}
\declaretheorem[
  name=Remark,
  Refname={Remark,Remarks}]{rmk}

\declaretheorem[
  name=Example,
  Refname={Example,Examples}]{ex}

\begin{document}

\title{\bf Distal Causal Excursion Effects: Modeling Long-Term Effects of Time-Varying Treatments in Micro-Randomized Trials}
\author{Tianchen Qian \\
    {\small Department of Statistics, University of California, Irvine}}
\date{}
\maketitle

\spacingset{1.4}
\begin{abstract}
    Micro-randomized trials (MRTs) play a crucial role in optimizing digital interventions. In an MRT, each participant is sequentially randomized among treatment options hundreds of times. While the interventions tested in MRTs target short-term behavioral responses (proximal outcomes), their ultimate goal is to drive long-term behavior change (distal outcomes). However, existing causal inference methods, such as the causal excursion effect, are limited to proximal outcomes, making it challenging to quantify the long-term impact of interventions. To address this gap, we introduce the distal causal excursion effect (DCEE), a novel estimand that quantifies the long-term effect of time-varying treatments. The DCEE contrasts distal outcomes under two excursion policies while marginalizing over most treatment assignments, enabling a parsimonious and interpretable causal model even with a large number of decision points. We propose two estimators for the DCEE---one with cross-fitting and one without---both robust to misspecification of the outcome model. We establish their asymptotic properties and validate their performance through simulations. We apply our method to the HeartSteps MRT to assess the impact of activity prompts on long-term habit formation. Our findings suggest that prompts delivered earlier in the study have a stronger long-term effect than those delivered later, underscoring the importance of intervention timing in behavior change. This work provides the critically needed toolkit for scientists working on digital interventions to assess long-term causal effects using MRT data.
\end{abstract}

\noindent%
\vfill

\newpage

\tableofcontents

\newpage

\spacingset{1.9} 

\section{Introduction}
\label{sec:introduction}

Micro-randomized trials (MRTs) are widely used to develop and optimize digital adaptive interventions \citep{klasnja2015microrandomized,liao2016sample}, with applications in health, education, and information systems \citep{walton2018optimizing,breitwieser2024realizing,pieper2024micro}. In an MRT, each participant is sequentially randomized among treatment options (e.g., receiving or not receiving a push notification encouraging behavioral change). These interventions directly target \textit{proximal outcomes} (short-term behavioral responses, such as step count in the next 30 minutes), but ultimately aim to influence \textit{distal outcomes} (long-term behaviors and habit formation).

To analyze MRT data, the causal excursion effect (CEE) was introduced to model the effect of time-varying treatments on proximal outcomes \citep{boruvka2018}. It has become the standard for primary and secondary analyses in MRTs \citep{klasnja2018efficacy,liu2023microrandomized}.  However, this estimand does not address distal outcomes, leaving key scientific questions unanswered. For example, in a physical activity study \citep{klasnja2015microrandomized}, do activity prompts lead to sustained habit formation? Current practice relies on behavioral theory to assume long-term effects, highlighting the need for statistical methods to infer them from MRT data.

To address this critical need, we propose the distal causal excursion effect (DCEE) to quantify the impact of time-varying treatments on a distal outcome measured at the end of an MRT. The DCEE contrasts distal outcomes under two treatment policies representing excursions (i.e., deviations) from the MRT policy. By marginalizing over the stochastic treatment assignments at most decision points, the DCEE enables a parsimonious causal effect model even with many decision points; a similar idea was used in CEE for proximal outcomes. This marginalization addresses a key challenge in MRT analysis: treatments are randomized at hundreds or thousands of decision points. Classical models such as marginal structural models and structural nested models \citep{robins2000marginalstructural} work well with a few treatment occasions but become impractical in MRTs with many decision points unless restrictive and often unrealistic assumptions are imposed \citep{rudolph2022estimation}.

To estimate the DCEE, we develop two estimators---one with cross-fitting and one without---both robust to outcome model misspecification. We establish their asymptotic properties and validate finite-sample performance through simulations. We apply the method to the HeartSteps MRT to assess the long-term effect of activity suggestions on habit formation at study end. Our analysis reveals that suggestions delivered earlier, despite being farther from the distal outcome, have a larger effect than those delivered later. We discuss the implications for designing future digital interventions.

\cref{sec:notation} defines the DCEE. \cref{sec:methods} presents the estimators and the asymptotic theory. \cref{sec:simulation} presents the simulation study. \cref{sec:application} presents the real data application. \cref{sec:discussion} concludes with discussion.

\section{DCEE Definition}
\label{sec:notation}

\subsection{Notation: MRT with a Distal Outcome}
\label{subsec:mrt-data-structure}

Consider an MRT with $n$ individuals, each in the trial for $T$ decision points where treatment is randomized. Variables without subscript $i$ refer to a generic individual. Let $A_t$ denote the binary treatment at decision point $t$: $A_t = 1$ for treatment, $A_t = 0$ for no treatment. Let $X_t$ denote observations recorded between decision points $t-1$ and $t$. The distal outcome $Y$, typically the primary health outcome, is measured at study end. The observed data trajectory for one individual is $O = (X_1,A_1,\ldots,X_T,A_T,Y)$. Overbars indicate sequences of variables; for example, $\bA_t = (A_1,\ldots,A_t)$. The history up to $t$ (excluding $A_t$) is $H_t = (X_1,A_1,\ldots,X_{t-1},A_{t-1}, X_t)$. At each $t$, $A_t$ is assigned with randomization probability $p_t(a|H_t) := P(A_t = a \mid H_t)$ for $a \in \{0,1\}$. Data from individuals are assumed independent and identically distributed samples from the distribution $P_0$, and expectations are taken under $P_0$ unless stated otherwise. \cref{fig:data-structure} illustrates the data structure.

\begin{figure}[tb]
    \centering  
        \begin{tikzpicture}[
        -Latex,auto,node distance =0.5 cm and 0.5 cm,semithick,
        state/.style ={circle, draw, minimum size = 0.7 cm, inner sep = 1pt},
        point/.style = {circle, draw, inner sep=0.04cm,fill,node contents={}},
        bidirected/.style={Latex-Latex,dashed},
        el/.style = {inner sep=2pt, align=left, sloped}]

        \node[state] (x1) at (0,0) {$X_1$};

        \node[state] (a1) [right =of x1] {$A_1$};

        \path (x1) edge (a1);


        \node[state] (x2) [right =0.7cm of a1] {$X_2$};
        \node[state] (a2) [right =of x2] {$A_2$};

        \path (x2) edge (a2);

        \path (a1) edge (x2);


        \node[draw=none] (dots1) [right =0.7cm of a2] {$\cdots$};
        \path [arrows = {-Computer Modern Rightarrow[line cap=round]}] (a2) edge[line width = 1mm] (dots1);

        \path (x1) edge[bend left=40] (x2);
        \path (x1) edge[bend left=50] (a2);
        \path (a1) edge[bend left=40] (a2);

        \node[state] (xt) [right =0.7cm of dots1] {$X_t$};
        \node[state] (at) [right =of xt] {$A_t$};

        \path (xt) edge (at);

        \path [arrows = {-Computer Modern Rightarrow[line cap=round]}] (dots1) edge[line width = 1mm] (xt);

        \node[draw=none] (dots2) [right =0.7cm of at] {$\cdots$};
        \path [arrows = {-Computer Modern Rightarrow[line cap=round]}] (at) edge[line width = 1mm] (dots2);

        \node[state] (xT) [right =0.7cm of dots2] {$X_T$};
        \node[state] (aT) [right =of xT] {$A_T$};

        \path (xT) edge (aT);

        \path [arrows = {-Computer Modern Rightarrow[line cap=round]}] (dots2) edge[line width = 1mm] (xT);

        \node[state] (y) [right =0.7cm of aT] {$Y$};
        \path [arrows = {-Computer Modern Rightarrow[line cap=round]}] (aT) edge[line width = 1mm] (y);

        \draw [decorate,decoration={brace,amplitude=3mm,mirror,raise=3mm}, -] (x1.south west) -- (xt.south east);

        \node[draw = none] (ht) [below = 0.5cm of a2] {$H_t$};
    \end{tikzpicture}
    \caption{MRT data structure with a distal outcome. To simplify notation, a thick arrow is used to denote arrows from all nodes on the left to all nodes on the right.}
    \label{fig:data-structure}
\end{figure}

In designing the MRT, researchers may deem it unsafe or unethical to deliver push notifications at certain times, such as when a participant is driving. At these decision points, the participant is considered ineligible for randomization, and no treatment is delivered. Formally, $X_t$ includes an indicator $I_t$, with $I_t = 1$ denoting being eligible for randomization at $t$, and $I_t = 0$ otherwise. If $I_t = 0$, then $A_t = 0$ deterministically. In the MRT literature, $I_t = 0$ is also called being unavailable for treatment \citep{boruvka2018}; here, we use ``ineligible'' for clarity. \revision{For simplicity, we sometimes explicitly write $I_t$ together with $H_t$ in conditional expectations, even though $I_t$ is part of $H_t$: for instance, $\EE(\cdot | H_t, I_t = 1)$ means $\EE(\cdot | H_t \setminus \{I_t\}, I_t = 1)$.}

Let $\PP_n$ denote the empirical mean over everyone. For any positive integer $k$, define $[k] := \{1, 2, \ldots, k\}$. The superscript $\star$ is used to indicate quantities associated with the true data-generating distribution $P_0$. For a vector $\alpha$ and a vector-valued function $f(\alpha)$, the notation $\partial_\alpha f(\alpha) := \partial f(\alpha) / \partial \alpha^T$ denotes the Jacobian matrix, where the $(i,j)$-th entry corresponds to the partial derivative of the $i$-th entry of $f$ with respect to the $j$-th entry of $\alpha$.

\subsection{Distal Causal Excursion Effect (DCEE)}
\label{subsec:def-cee}

To define the causal effect, we use the potential outcomes notation \citep{rubin1974estimating,robins1986new}. Lowercase letters denote instantiations (non-random values) of corresponding random variables; for example, $a_t$ is an instantiation of $A_t$. For each individual, let $X_t(\ba_{t-1})$ be the value of $X_t$ that would have been observed at $t$ if the individual were assigned treatment sequence $\ba_{t-1}$. The potential outcome of $H_t$ under $\ba_{t-1}$ is $H_t(\ba_{t-1}) = \{X_1, a_1, X_2(a_1), a_2, \ldots, X_t(\ba_{t-1})\}$. The potential outcome of $Y$ under $\ba_T$ is $Y(\ba_T)$.

We use the term \textit{policy} to refer to any decision rule (static or dynamic) for assigning $A_t$ given the history $H_t$ for $t \in [T]$ \citep{murphy2003optimal}. The MRT policy, denoted by $D := (d_1,\ldots, d_T)$, is the rule for stochastically assign $A_t$ at each $t$ in the MRT. Specifically, $d_t(\cdot)$ is a stochastic mapping from $H_t$ to $\{0,1\}$, where $d_t(H_t) = a$ with probability $P(A_t = a \mid H_t)$ for $a \in \{0,1\}$. For each $t$, we consider two alternative decision rules for defining the causal effect: $d_t^1$, which always assigns $A_t = 1$ unless $I_t = 0$ (in which case $A_t = 0$), and $d_t^0$, which always assigns $A_t = 0$. In other words, $d_t^1(H_t) = I_t$, and $d_t^0(H_t) = 0$. For $a\in\{0,1\}$, let $D_{d_t = d_t^a}$ denote the policy obtained by replacing $d_t$ with $d_t^a$ in the MRT policy $D$:
\begin{align}
    D_{d_t = d_t^a} := (d_1,\ldots,d_{t-1},d_t^a,d_{t+1},\ldots,d_T). \label{eq:def-alternative-policy}
\end{align}
$D_{d_t = d_t^1}$ and $D_{d_t = d_t^0}$ are two \textit{excursions} \citep{guo2021discussion} from the MRT policy $D$: they deviate from $D$ at decision point $t$ to always assign $A_t = 1$ (while respecting eligibility) or always assign $A_t = 0$, respectively.

Let $S_t$ denote a subset of $H_t$, representing moderators of interest. For example, setting $S_t = \emptyset$ defines a fully marginal effect averaging over all moderators; setting $S_t = A_{t-1}$ captures effect moderation by previous intervention; and setting $S_t = X_t$ captures effect moderation by current covariates. Note that $S_t$ can be time-varying. We define the \textit{distal causal excursion effect} (DCEE) of $A_t$ on $Y$ as
\begin{align}
    \tau(t,s) := \EE \Big\{ Y(D_{d_t = d_t^1}) - Y(D_{d_t = d_t^0}) \mid S_t(\bd_{t-1}) = s \Big\}. \label{eq:def-cee}
\end{align}
$\tau(t,s)$ captures the difference in the potential outcomes of $Y$ under two excursion policies, $D_{d_t = d_t^1}$ and $D_{d_t = d_t^0}$, conditional on $S_t$. In other words, $\tau(t,s)$ contrasts what would have happened if a participant followed the MRT policy throughout the study except at decision point $t$, where they received $A_t = I_t$ versus $A_t = 0$. \revision{The way $I_t$ enters the definition here differs from its role in defining (proximal) causal excursion effects \citep{boruvka2018}; we elaborate on this distinction further in \cref{rmk:dcee_vs_cee}.}

Below we present several examples with different data-generating distributions to illustrate the DCEE and highlight its scientific relevance for digital interventions. In all examples, we assume larger $Y$ is desirable, the error term $\epsilon$ is exogenous with mean 0, and set $S_t = \emptyset$ to focus on fully marginal effects $\tau(t)$. A time-varying variable is exogenous if it is independent of its own history and the history of other variables. More general versions of these examples, along with detailed derivations, are in Supplementary Material \ref{A-sec:dcee-examples}. Importantly, the estimators proposed in \cref{sec:methods} remain valid as long as the assumptions in \cref{thm:normality} hold; they do not require knowledge of most aspects of the data-generating process, such as how $Y$ depends on $A_t$ and $X_t$ or how $A_t$ influences $X_{t+1}$ and future eligibility $I_{t+1}$.

\begin{ex}[DCEE marginalizes over effect modifiers not in $S_t$]
    \label{ex:marginalize-over-effect-modifiers}
    \normalfont
    Suppose that for all $t\in[T]$, the covariate $X_t$ is exogenous, all individuals are always eligible ($I_t \equiv 1$), and treatment $A_t \sim \bern(p)$ is exogenous. Suppose $Y = g(\bX_T) + \sum_{t=1}^T A_t (\alpha_t + \beta_t X_t) + \epsilon$ for some unspecified $g$ and $E(\epsilon) = 0$. Since we set $S_t = \emptyset$, the DCEE $\tau(t)$ represents the fully marginal effect of $A_t$, given by $\tau(t) = \alpha_t + \beta_t ~\EE(X_t)$. This illustrates how the DCEE naturally marginalizes over effect modifiers not explicitly included in $S_t$.
\end{ex}

\begin{ex}[DCEE captures user burden]
    \label{ex:treatment-burden}
    \normalfont
    Consider the same data-generating process as \cref{ex:marginalize-over-effect-modifiers}, but instead suppose that there is a negative interaction between consecutive treatments: $Y = g(\bX_T) + \sum_{t=1}^T \beta_t A_t - \sum_{t=1}^{T-1} \alpha_t A_t A_{t+1} + \epsilon$ with $\alpha_t \geq 0$. Here, treatment effects diminish when treatments are delivered in succession, capturing \textit{user burden}. The average-over-time effect is $\frac{1}{T}\sum_{t=1}^T \tau(t) = \frac{1}{T} \sum_{t=1}^T \{\beta_t - (\alpha_t + \alpha_{t-1})p \}$, where $\alpha_0 := 0$. The term $-(\alpha_t + \alpha_{t-1})p$ reflects the deterioration of treatment effects due to user burden, which increases with the treatment probability $p$. Thus, the DCEE automatically accounts for user burden, and as a result different MRTs with different $p$ will yield different DCEE.
\end{ex}

\begin{ex}[DCEE combines direct and indirect effects]
    \label{ex:direct-indirect-effects}
    \normalfont
    Consider a two-timepoint setting ($T=2$), where $X_1$ is exogenous, all individuals are always eligible ($I_1 = I_2 \equiv 1$), treatments $A_1, A_2 \sim \bern(p)$ are exogenous, and the treatment $A_1$ influences the time-varying covariate $X_2$ through $X_2 \mid A_1 \sim \bern(\rho_0 + \rho_1 A_1)$. Suppose $Y = \gamma_0 + \gamma_1 X_1 + \gamma_2 X_2 + \beta_1 A_1 + \beta_2 A_2 + \epsilon$. In this case, the DCEE at $t=1$ is $\tau(1) = \beta_1 + \gamma_2 \rho_1$. This expression captures both the direct effect of $A_1$ on $Y$ ($\beta_1$) and the indirect effect mediated through $X_2$ ($\gamma_2 \rho_1$).
\end{ex}

\begin{ex}[DCEE captures treatment effects on future eligibility]
    \label{ex:eligibility}
    \normalfont
    Consider a two-timepoint setting ($T=2$), where all individuals are eligible at $t=1$ ($I_1 \equiv 1$), treatment at $t=1$ influences future eligibility: $I_2 \mid A_1 \sim \bern(\rho_0 - \rho_1 A_1)$, and the treatment assignments follow $A_1 \sim \bern(p)$ and $A_2 \mid I_2 = 1 \sim \bern(p)$. Suppose $Y = \beta_0 + \beta_1 A_1 + \beta_2 A_2 + \epsilon$. Assume $\beta_1, \beta_2, \rho_0, \rho_1 > 0$. In this case, the DCEE at $t=1$ is $\tau(1) = \beta_1 - \beta_2 p \rho_1$. Here, the term $-\beta_2 p \rho_1$ captures the reduced effect of $A_1$ on $Y$ due to its negative impact on future eligibility.  This demonstrates how the DCEE accounts for delayed effects of treatment via eligibility constraints, a key consideration in MRTs as eligibility may depend on past treatment history.
\end{ex}

\begin{rmk}[Differences between DCEE and CEE]
    \label{rmk:dcee_vs_cee}
    \normalfont
    The DCEE differs from the causal excursion effect (CEE) \citep{boruvka2018} in several key ways: (i) Outcome focus: DCEE focuses on the end-of-study distal outcome, whereas CEE focuses on short-term proximal outcomes. (ii) Policy comparison: DCEE contrasts treatment policies on the entire sequence $A_1,\ldots,A_T$, while CEE contrasts policies up to time $t$ (or up to $A_{t+\Delta-1}$ for some prespecified $\Delta \geq 1$). (iii) Eligibility handling: DCEE marginalizes over eligibility and compares policies that respect eligibility constraints, whereas CEE conditions on eligibility at a given decision point. This last distinction is especially important: DCEE accounts for the long-term consequences of treatment on future eligibility. Thus, a treatment that is effective but substantially reduces future eligibility may yield an attenuated DCEE (see \cref{ex:eligibility}). This issue does not arise with CEE, which focuses only on short-term outcomes. \revision{Consequently, one cannot simply repurpose CEE estimation methods (e.g., weighted and centered least squares \citep{boruvka2018}) to estimate the DCEE by treating the distal outcome as a proximal outcome at each decision point. Such an approach would fail to account for eligibility probability, thus resulting in biased DCEE estimates. This is illustrated in our simulation results (\cref{sec:simulation}).}   \qed
\end{rmk}

\begin{rmk}[Differences between DCEE and other causal estimands]
    \normalfont
    The DCEE differs from marginal structural models (MSMs) and structural nested mean models (SNMMs) \citep{robins2000marginalstructural}. In an MRT, a MSM models $Y$ under fixed policies, and a SNMM models a treatment blip effect assuming future treatments are set to zero. Without strong (often unrealistic) assumptions, both require more parameters than MRT data can support. In contrast, DCEE compares policies that deviate at a single decision point from the MRT policy, thus enabling parsimonious, interpretable modeling.  \qed
\end{rmk}

\begin{rmk}[Differences between DCEE and dynamic treatment regime]
    \normalfont
    \revision{The DCEE differs from dynamic treatment regimes (DTRs) \citep{murphy2003optimal,robins2004optimal} in several aspects. First, DTR approaches contrast entire sequences of decision rules (full dynamic policies), whereas DCEE compares ``excursion'' policies that only differ at a single decision point and match the reference policy elsewhere. Second, the DTR framework aims to optimize distal outcomes by identifying an optimal policy, while DCEE focuses on evaluating the long-term effect of a local treatment decision without optimizing across sequences. Third, DTR is most suitable for designs with few decision points (e.g., sequential multiple assignment randomized trials \citep{murphy2005smart}), while DCEE is suitable for MRTs with many decision points. These distinctions reflect a bias-variance trade-off: DCEE defines a more parsimonious estimand suited for intensive longitudinal settings.}
\end{rmk}

\subsection{Causal Assumptions and Identification}
\label{subsec:causal-assumptions-and-identification}

We make the following standard causal assumptions regarding data from an MRT.
\begin{asu}
    \label{asu:causal-assumptions}
    \normalfont
    \spacingset{1.5}
    \begin{asulist}
        \item \label{asu:consistency} (SUTVA.) There is no interference across individuals and the observed data equals the potential outcome under the observed treatment. As a result, $X_t = X_t(\bA_{t-1})$ for $t\in[T]$ and $Y = Y(\bA_T)$.
        \item \label{asu:positivity} (Positivity.) There exists a positive constant $\tau > 0$, such that if $P(H_t = h_t, I_t = 1) > 0$ then $\tau < P(A_t = a \mid H_t = h_t, I_t = 1) < 1-\tau$ for $a \in \{0, 1\}$.
        \item \label{asu:sequential-ignorability} (Sequential ignorability.) For $1 \leq t \leq T$, the potential outcomes $\big\{X_{t+1}(\ba_t)$, $X_{t+2}(\ba_{t+1})$, $\ldots, X_T(\ba_{T-1}), Y(\ba_T): \ba_T \in \{0,1\} \times \{0,1\} \times \cdots \times \{0,1\} \big\}$ are conditionally independent of $A_t$ given $H_t$.
    \end{asulist}
\end{asu}
Positivity and sequential ignorability are guaranteed by the MRT design. SUTVA holds for most MRTs without the social aspect but is violated under interference, where one participant's treatment affects another's outcome. In such cases, frameworks incorporating causal interference are needed \citep{hudgens2008toward,shi2022assessing}, which we do not consider here. In Supplementary Material \ref{A-sec:proof-identification}, we prove the following identification result for $\tau(t,s)$.
\revision{\begin{thm}
    \label{thm:identification}
    Under \cref{asu:causal-assumptions}, we have
    \begin{align}
        \tau(t,s) 
        & = \EE \bigg\{ \frac{\indic(A_t = I_t)}{P(A_t = I_t \mid H_t)} Y - \frac{\indic(A_t = 0)}{P(A_t = 0 \mid H_t)} Y ~\bigg|~ S_t = s \bigg\}. \label{eq:identify-cee-ipw} \\
        & = \EE \Big\{ \EE (Y \mid H_t, A_t = I_t ) - \EE (Y \mid H_t, A_t = 0 ) \mid S_t = s \Big\} \label{eq:identify-cee-ie} \\
        & = \EE(I_t) \EE \Big\{ \EE(Y \mid H_t, I_t = 1, A_t = 1) - \EE(Y \mid H_t, I_t = 1, A_t = 0) \mid S_t = s, I_t = 1 \Big\}. \label{eq:identify-cee-ie2}        
    \end{align}
\end{thm}}

\section{Robust Estimators for DCEE}
\label{sec:methods}

Let $f(t,s)$ be a pre-specified $p$-dimensional feature vector that depends on the decision point index and the effect modifier value. We consider estimating the best linear projection of $\tau(t,s)$ on $f(t,s)$, averaged over decision points. Specifically, the true parameter $\beta^\star \in \RR^p$ is defined as
\begin{align}
    \beta^\star = \arg\min_{\beta \in \RR^p} \sum_{t=1}^T \omega(t) \EE\Big[\{ \tau(t,S_t) - f(t,S_t)^T \beta \}^2\Big], \label{eq:beta-star-def}
\end{align}
where $\omega(t)$ is a pre-specified weight function with $\sum_{t=1}^T \omega(t) = 1$. This is similar to linearly parameterizing $\tau(t,S_t) = f(t,S_t)^T \beta^\star$, but by defining $\beta^\star$ this way ensures interpretability even when the linear model is misspecified.

Researchers can choose different $f(t,s)$ and $\omega(t)$ based on the scientific question. For example, when $S_t = \emptyset$, setting $f(t,s) = (1,t)$ or $(1,t, t^2)$ models how the DCEE varies linearly or quadratically over the decision point where the excursion takes place. Basis functions of $t$ can also be included for flexibility (see \cref{sec:application}). When $S_t \neq \emptyset$, $f(t,s)$ can be additive in $t,s$ or include interactions. For weights $\omega(t)$, setting $\omega(t) = 1/T$ equally weights all decision points in estimating $\beta$. On the other hand, if the goal is to estimate $\tau(t,s)$ at a particular $t = t_0$, one can set $\omega(t_0) = 1$ and $\omega(t) = 0$ for all $t \neq t_0$.

\begin{rmk}[$\beta^\star$ as a weighted average]
    \label{rmk:beta-as-weighted-average}
    \normalfont
    When $\omega(t) = 1/T$, $S_t = \emptyset$, and $f(t,s) = 1$ for all $t \in [T]$, $\beta^\star$ can be shown to adopt the following form:
    \revision{\begin{align}
        \beta^\star = \frac{1}{T} \sum_{t=1}^T \EE(I_t) \EE \Big\{ \EE(Y | H_t, I_t = 1, A_t = 1) - \EE(Y | H_t, I_t = 1, A_t = 0) \Big\}. \label{eq:beta-star-projection-1}
    \end{align}}
    This is a time-averaged, fully marginal (i.e., not conditional on any effect modifiers) excursion effect of $A_t$ on $Y$ when the participant is eligible, discounted by the probability of being eligible at each decision point. For general $S_t$ and $f(t,s)$, when $\omega(t) = 1/T$, $\beta^\star$ can be interpreted as a weighted average of the moderated excursion effects:
    \begin{align}
        \beta^\star = & \bigg[\sum_{t=1}^T \EE \big\{f(t,S_t)f(t,S_t)^T \big\} \bigg]^{-1} \nonumber\\
        & \times \sum_{t=1}^T \EE(I_t) \EE \Big[ \big\{ \EE(Y | H_t, A_t = 1, I_t = 1) - \EE(Y | H_t, A_t = 1, I_t = 0) \big\} f(t,S_t) \Big]. \label{eq:beta-star-projection-2}
    \end{align}
    The weighted average form of $\beta^\star$ for more general settings is provided and proved in the Supplementary Material \ref{A-sec:beta-star-projection}, which immediately implies \eqref{eq:beta-star-projection-1} and \eqref{eq:beta-star-projection-2}. \qed
\end{rmk}

The identification equation \eqref{eq:identify-cee-ipw} motivates a preliminary estimating function for $\beta$:
\revision{\begin{align}
    \xi(\beta) := \sum_{t=1}^T \omega(t) \bigg\{ \frac{\indic(A_t = I_t)}{P(A_t = I_t \mid H_t)} Y - \frac{\indic(A_t = 0)}{P(A_t = 0 \mid H_t)} Y - f(t, S_t)^T\beta \bigg\} f(t,S_t). \label{eq:est-fcn-pre-proj}
\end{align}}
We further subtract from $\xi(\beta)$ its projection on the score functions of the treatment selection probabilities to obtain a more efficient estimating function \citep{robins1999testing}:
\begin{align}
    \phi(\beta, \mu) := \sum_{t=1}^T \omega(t) \bigg[ I_t \frac{(-1)^{1-A_t}}{p_t(A_t \mid H_t)} \Big\{ Y - p_t(0\mid H_t) \mu_t(H_t, 1) - p_t(1\mid H_t) \mu_t(H_t, 0) \Big\} - f(t, S_t)^T\beta \bigg] f(t,S_t). \label{eq:est-fcn-post-proj}
\end{align}
Here, $\mu = (\mu_1, \ldots, \mu_T)$ is an infinite-dimensional nuisance parameter with truth $\mu_t^\star(H_t, A_t) = \EE(Y \mid H_t, A_t)$. A detailed derivation of \eqref{eq:est-fcn-post-proj} is in the Supplementary Material \ref{A-sec:est-fcn-post-proj-derivation}.

The form of $\phi(\beta, \mu)$ motivates two estimators, $\widehat\beta$ and $\widetilde\beta$, as depicted in Algorithms \ref{algo:estimator-ncf} and \ref{algo:estimator-cf}. $\widehat\beta$ does not use cross-fitting and $\widetilde\beta$ uses cross-fitting \citep{chernozhukov2018double}.

\begin{algorithm}[htbp]
    \caption{A two-stage estimator $\widehat\beta$ (without cross-fitting)}
    \label{algo:estimator-ncf}
    \spacingset{1.5}
    \vspace{0.3em}
    \textbf{Stage 1:} Fit $\EE (Y \mid H_t, A_t)$ for $t \in [T]$. Denote the fitted model by $\widehat\mu_t(H_t, A_t)$. Denote $\widehat\mu := (\widehat\mu_1,\ldots,\widehat\mu_T)$.

    \textbf{Stage 2:} Obtain $\widehat\beta$ by solving $\PP_n \phi(\beta, \widehat\mu) = 0$ with $\phi$ defined in \eqref{eq:est-fcn-post-proj}.
    \vspace{0.3em}
\end{algorithm}

\begin{algorithm}[htbp]
    \caption{A two-stage estimator $\widetilde\beta$ (with cross-fitting)}
    \label{algo:estimator-cf}
    \spacingset{1.5}
    \vspace{0.3em}
    \textbf{Stage 1:} Take a $K$-fold equally-sized random partition $(B_k)_{k=1}^K$ of observation indices $[n] = \{1,\ldots,n\}$. Define $B_k^c = [n] \setminus B_k$ for $k \in [K]$. For each $k \in [K]$, use solely observations from $B_k^c$ and fit $\EE (Y \mid H_t, A_t)$ for $t \in [T]$. The fitted models using $B_k^c$ are denoted by $\widehat\mu_{kt}$, and let $\widehat\mu_k := (\widehat\mu_{k1},\ldots,\widehat\mu_{kT})$.

    \textbf{Stage 2:} Obtain $\widetilde\beta$ by solving $K^{-1} \sum_{k=1}^K \PP_{n,k} \phi(\beta, \widehat\mu_k) = 0$. Here $\PP_{n,k}$ denotes empirical average over observations from $B_k$.
    \vspace{0.3em}
\end{algorithm}

We establish the asymptotic normality for $\widehat\beta$ and $\widetilde\beta$ in \cref{thm:normality}, which is proved in the Supplementary Material \ref{A-sec:proof-normality}.

\begin{thm}[Asymptotic normality of $\widehat\beta$ and $\widetilde\beta$]
    \label{thm:normality}
    \spacingset{1.5}
    Suppose \cref{asu:causal-assumptions} hold and consider $\beta^\star$ defined in \eqref{eq:beta-star-def}. Suppose $\widehat\mu$ converges to some limit $\mu'$ (not necessarily the true $\mu^\star$) in $L_2$. Under regularity conditions, we have
    \begin{align}
        \sqrt{n}(\widehat\beta - \beta^\star) \dto N (0, V) \text{ as } n\to\infty, \label{eq:asymptotic-normality}
    \end{align}
    with $V := \EE\{ \partial_\beta \phi(\beta^\star, \mu') \}^{-1} \EE [ \{ \phi(\beta^\star, \mu') \phi(\beta^\star, \mu')^T \} ] \EE\{ \partial_\beta \phi(\beta^\star, \mu') \}^{-1,T}$. $V$ can be consistently estimated by $\PP_n\{ \partial_\beta \phi(\widehat\beta, \widehat\mu) \}^{-1} \PP_n [ \{ \phi(\widehat\beta, \widehat\mu) \phi(\widehat\beta, \widehat\mu)^T \} ] \PP_n\{ \partial_\beta \phi(\widehat\beta, \widehat\mu) \}^{-1,T}$.

    For the estimator $\widetilde\beta$ that uses cross-fitting, assume that $\widehat\mu_k$ converges to some limit $\mu'$ (not necessarily the true $\mu^\star$) in $L_2$ for each $k\in[K]$. Under regularity conditions, the asymptotic normality \eqref{eq:asymptotic-normality} holds with $\widehat\beta$ replaced by $\widetilde\beta$, in which case $V$ can be consistently estimated by 
    \begin{align*}
        \bigg[\frac{1}{K}\sum_{k=1}^K \PP_{n,k}\big\{ \partial_\beta \phi(\widehat\beta, \widehat\mu_k) \big\} \bigg]^{-1} \bigg[\frac{1}{K}\sum_{k=1}^K \PP_{n,k} \big\{ \phi(\widehat\beta, \widehat\mu_k) \phi(\widehat\beta, \widehat\mu_k)^T \big\} \bigg]
        \bigg[\frac{1}{K}\sum_{k=1}^K \PP_{n,k}\big\{ \partial_\beta \phi(\widehat\beta, \widehat\mu_k) \big\} \bigg]^{-1,T},
    \end{align*}
    where $\PP_{n,k}$ denotes the empirical average over observations from $B_k$.
\end{thm}

\begin{rmk}[Robustness]
    \label{rmk:estimator-robustness}
    \normalfont
    $\widehat\beta$ and $\widetilde\beta$ are robust in the sense that their consistency and asymptotical normality hold even with arbitrarily fitted nuisance parameter $\widehat\mu$, which need not converge to the true $\mu^\star$. Furthermore, the asymptotic variance $V$ involves only $\mu'$, the limit of $\widehat\mu$, and does not depend on how $\widehat\mu$ is fitted. This is because the estimating function $\phi(\beta,\mu)$ is globally robust \citep{cheng2023efficient}, as we prove in Supplementary Material \ref{A-sec:proof-normality}. Of course, a better $\widehat\mu$ improves the efficiency of $\widehat\beta$ and $\widetilde\beta$. Using nuisance parameters not required for identification to improve efficiency is common in causal inference \citep[e.g.,][]{tsiatis2008covariate,lok2024estimating}. \qed
\end{rmk}

\section{Simulation}
\label{sec:simulation}

\subsection{Generative Model and True Parameter Value} 
\label{subsec:simulation-dgm}

We construct a generative model that incorporates common complications in actual MRTs: endogenous time-varying covariates influenced by past treatments and affecting future treatments, effect moderation by time-varying covariates, interactions between current and past treatments that reflect user burden, and eligibility constraints.

In the generative model, each individual has $T = 30$ decision points, and data across individuals are i.i.d. For a generic individual, $X_t$ is a time-varying continuous covariate that depends on $A_{t-1}$ and $X_{t-1}$: $X_t = \theta_0 + \theta_1 A_{t-1} + \theta_2 X_{t-1} + \eta_t$, with $\eta_t \sim N(0,1)$ exogenous. $Z_t$ is a binary time-varying covariate that depends on $A_{t-1}$ and $Z_{t-1}$: $Z_t \sim \text{Bernoulli}\{\expit(\zeta_0 + \zeta_1 A_{t-1} + \zeta_2 Z_{t-1})\}$. We set $X_0 = Z_0 = A_0 = 0$. The eligibility indicator is exogenous: $I_t \sim \text{Bernoulli}(0.8)$. The randomization probability is $P(A_t = 1 \mid H_t, I_t = 1) = \expit\{(t-T/2) / T + Z_t - 0.5 + X_t / 6\}$, ensuring comparable influence of each variable on $A_t$. The outcome is
\begin{align*}
    Y = \sum_{t=1}^T \xi_t \{g(X_t/12 + 0.5) + Z_t\} + \sum_{t=1}^T A_t(\alpha_t + \nu_t X_t + \gamma_t Z_t + \lambda_t A_{t-1}) + \epsilon,
\end{align*}
where $g(\cdot)$ is the $\text{Beta}(2,2)$ density function (nonlinear), \revision{and $\epsilon \sim N(0,1)$ is exogeneous}. The observed data for an individual is $\{X_t,Z_t,I_t,A_t:1 \leq t \leq T\} \cup \{Y\}$, with $P(A_t = 1 \mid H_t, I_t = 1)$ is known from the trial design.

The parameter values are set as follows: $\theta_0 = -0.5$, $\theta_1 = \theta_2 = 0.5$; $\zeta_0 = -1$, $\zeta_1 = \zeta_2 = 1$; $\alpha_t = 1 + 2(t-1)/(T-1)$; $\nu_t = 1 + (t-1)/(T-1)$; $\gamma_t = 1 + 0.5(t-1)/(T-1)$; $\lambda_t = -1 - (t-1)/(T-1)$; $\xi_t = 1 + (t-1)/(T-1)$. The negative $\lambda_t$ reflects user burden (\cref{ex:treatment-burden}). The magnitudes of $\alpha_t, \nu_t, \gamma_t, \lambda_t, \xi_t$ all increase with $t$, characterizing a scenario where variables closer to the distal outcome $Y$ have a larger impact on $Y$.

We consider two sets of estimands and set $\omega(t) = 1/T$. The first estimand corresponds to the fully marginal DCEE by setting $S_t = \emptyset$ in \eqref{eq:def-cee} and $f(t,S_t) = 1$ in \eqref{eq:beta-star-def}, and the true parameter value is $\beta^\star_1 = 1.603$. The second set of estimands corresponds to the DCEE moderated by the binary covariate $Z_t$, i.e., through setting $S_t = Z_t$ and $f(t,S_t) = (1, Z_t)^T$. The true parameter values are $(\beta^\star_2, \beta^\star_3) = (1.207, 0.881)$. The numerical computation of $\beta^\star_1$ and $(\beta^\star_2, \beta^\star_3)$ is detailed in Supplementary Material \ref{A-sec:simulation-true-parameter}.


\subsection{Proposed Estimator and Comparator Methods} 
\label{subsec:simulation-estimators}

We consider the proposed estimators $\widehat\beta$ (without cross-fitting) and $\widetilde\beta$ (with $K=5$ fold cross-fitting). In both estimators, the working models for nuisance parameters $\mu_t(H_t,1)$ and $\mu_t(H_t,0)$ are pooled over $t \in [T]$ and fitted using either the generalized additive model (\texttt{gam} in R package \texttt{mgcv} \citep{wood2011fast}) \revision{or a fast implementation of random forests (R package \texttt{ranger} \citep{wright2017ranger})}. Each working model includes covariates $X_t$ (as a penalized cubic spline in \texttt{gam}) and untransformed $Z_t$. These working models are misspecified due to omitted history variables (because $\mu_t$ depends on the full history in addition to $X_t, Z_t$) and incorrect functional form (the dependence of $\mu_t$ on $X_t$ and $Z_t$ varies with $t$, which is not captured in the working model that pools over all $t\in[T]$). We introduce this misspecification to illustrate the robustness of the proposed estimator to model specificiation in $\widehat\mu$ (\cref{rmk:estimator-robustness}).


The DCEE is a novel estimand, with no existing method directly applicable. Nonetheless, we include as comparators three common approaches for analyzing longitudinal data with time-varying treatments, to illustrate the advantage of our method: the generalized estimating equations \citep[GEE, ][]{liang1986longitudinal}, the structural nested mean model \citep[SNMM, ][]{robins1994snmm}, \revision{and the weighted and centered least squares \citep[WCLS, ][]{boruvka2018}}. To use GEE for estimating the marginal DCEE $\beta^\star_1$, we specify the mean model for $Y$ as a linear combination of $A_t, X_t, Z_t$, pooling data over all $t\in[T]$; the coefficient for $A_t$ (same across all $t$) is the GEE estimator for $\beta^\star_1$. For the moderated DCEE $\beta^\star_2$ and $\beta^\star_3$, we specify the mean model for $Y$ as a linear combination of $A_t, X_t, Z_t, A_t \times Z_t$, again pooling over $t$. The coefficients for $A_t$ and $A_t \times Z_t$ are the GEE estimators for $\beta^\star_2$ and $\beta^\star_3$, respectively. We use a working independence correlation matrix and the robust standard error estimator.

For SNMM, we used the g-estimation implementation in the R package \texttt{DTRreg} \citep{wallace2017dtrreg}. To estimate the marginal DCEE $\beta^\star_1$, we specified a blip model with intercept only for each $t \in [T]$. The treatment assignment model included $X_t$ and $Z_t$ and correctly captured time dependence because \texttt{DTRreg} fits separate models for each $t$. The treatment-free outcome model was specified as a linear model depending on $X_t$ and $Z_t$, mimicking our nuisance model for $\mu_t$ to ensure comparability. For estimating the moderated DCEE parameters $\beta^\star_2$ and $\beta^\star_3$, we additionally included $Z_t$ in the blip model.

\revision{For WCLS, we used the \texttt{wcls} function in R package \texttt{MRTAnalysis}. We set the proximal outcome argument as the distal outcome (same value repeated for each decision point), and used as control variables a linear form of $X_t$ and $Z_t$. The moderator argument is set to intercept only (for estimating $\beta^\star_1$) or include $Z_t$ (for estimating $\beta^\star_2$ and $\beta^\star_3$). The output proximal CEE are treated as the estimated DCEE in the simulation.}

\subsection{Results} 
\label{subsec:simulation-model-specification}

We conducted simulations for sample sizes $n = 30, 50, 100, 200, 300, 400, 500$, each with 1000 replicates. Results are shown in \cref{fig:simulation-result}. The proposed DCEE estimators perform well even at $n=30$: for all estimands, bias is near zero and 95\% confidence interval coverage is close to nominal. There is no substantial difference between versions with and without cross-fitting. In contrast, the comparator methods GEE, SNMM, \revision{and WCLS} show substantial bias and poor coverage, because they were not designed for estimating the DCEE.

\begin{figure}[htbp]
    \centering
    \includegraphics[width = \textwidth]{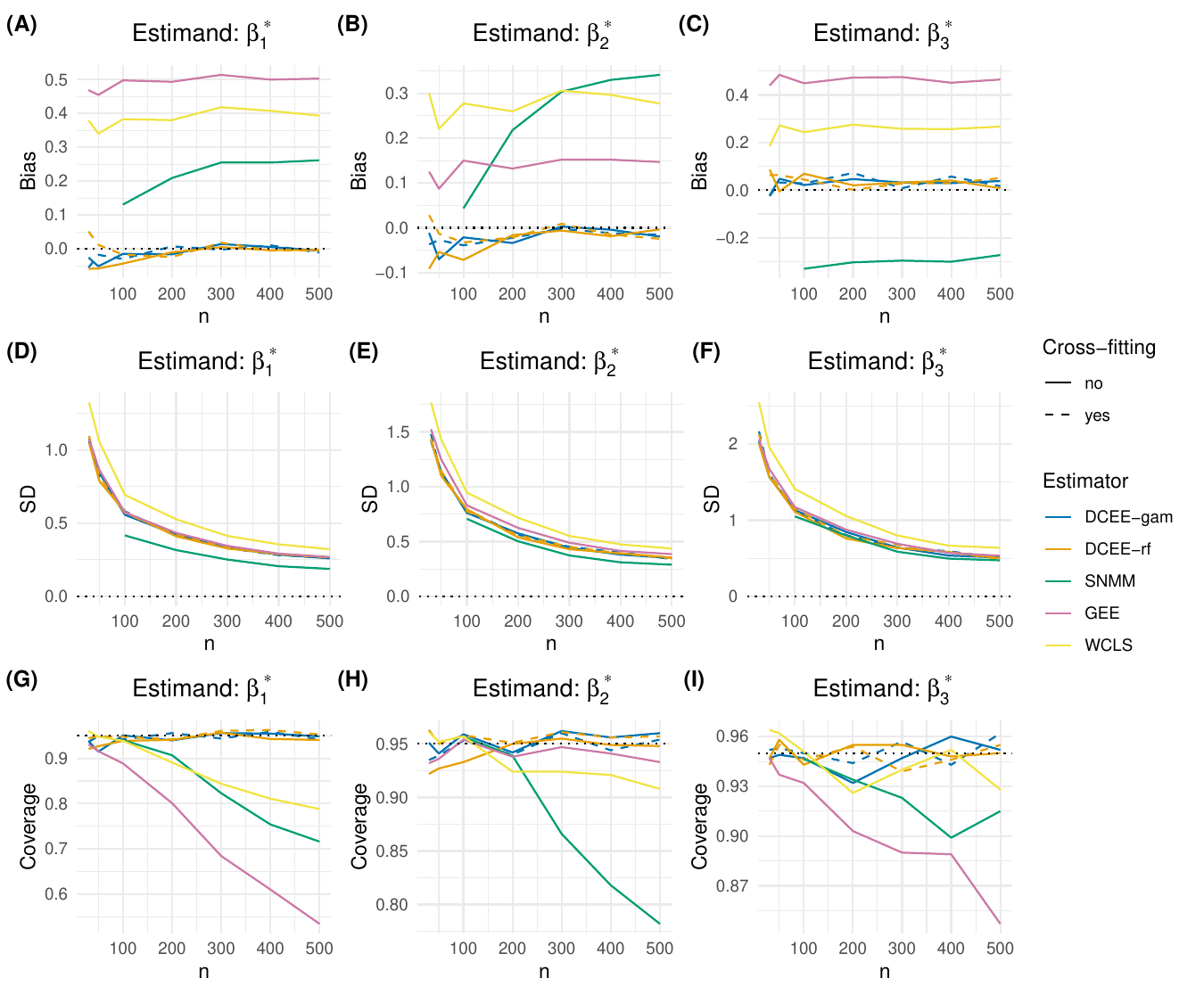}
    \caption{The numerical performance of DCEE (using GAM or random forests for nuisance parameter estimation, and with or without cross-fitting) GEE, SNMM, and WCLS for the three estimands in terms of bias (plots A--C), standard deviation (plots D--F), and the coverage of 95\% confidence intervals (plots G--I).}
    \label{fig:simulation-result}
\end{figure}

\revision{We observed very similar finite-sample performance between the cross-fitted and non-cross-fitted versions of the proposed estimator when the nuisance outcome model was estimated using either GAM or random forests. This may reflect that the nuisance functions in our generative setting were not overly complex and could be reliably learned by these algorithms. Given that our estimating equation is globally robust and does not rely on the nuisance model for identification, both approaches perform similarly in this context. A more comprehensive investigation of their relative merits is an important direction for future work. In practice, we recommend using cross-fitting when employing highly flexible learners in more complex settings or when overfitting bias is a concern \citep{zivich2021crossfit,naimi2021challenges}, but using the non-cross-fitted version if cross-fitting fails to converge (e.g., due to very small sample sizes).}

\section{Application}
\label{sec:application}

We use data from HeartSteps I, the first in a sequence of HeartSteps MRTs, designed to support sedentary adults in achieving and sustaining recommended physical activity levels \citep{klasnja2015microrandomized}. In this application, we focus on the micro-randomized activity suggestions. $n = 37$ participants were enrolled for 6 weeks, with decision points five times per day at pre-specified times, totaling $T = 210$ per person. At each decision point, a participant was eligible to be randomized if not driving or walking and if their phone had a stable internet connection. At every eligible decision point, a participant was randomized with probability 0.6 to receive an activity suggestion (a push notification suggesting brief walking or stretching) and probability 0.4 to receive no intervention. Across all participants, approximately 80\% of the decision points were eligible. Each participant wore a wristband tracker that continuously recorded step count.

The long-term goal of these activity suggestions was to help participants develop a habit of being physically active. Therefore, in this illustration we use the average daily step count in Week 6 as the distal outcome, serving as a proxy for habit formation. We analyze the effect of suggestions delivered in the first 5 weeks on this distal outcome. We discuss the distal outcome choice at the end of this section.

We set $\omega(t) = 1/T$. To fit nuisance parameters $\mu_t(H_t,1)$ and $\mu_t(H_t,0)$, we use the generalized additive model, pooling data over all $t \in [T]$, and include the step count in the prior 30 minutes and a location indicator (home/work $= 1$, elsewhere $= 0$).

We first examine the fully marginal excursion effect by setting $S_t = \emptyset$ and $f(t,s) = 1$ in \eqref{eq:beta-star-def}. The estimated effect is $\widehat\beta = 69$ (95\% CI $=[-107,245]$), which indicates that, averaged over Weeks 1--5, sending an activity suggestion has no detectable effect on Week 6 step count. We then assess the effect moderation by the decision point index $t$ (proportional to time in study), motivated by prior evidence of a diminishing effect over time on a short-term proximal outcome measured after each decision point \citep{klasnja2018efficacy}. We set $S_t = \emptyset$, $f(t,s) = (1,t-1)^T$, and $\beta = (\beta_0, \beta_1)^T$, where the $t-1$ ensures that $\beta_0$ reflects the effect at the first decision point. We find $\widehat\beta_0 = 344$ (95\% CI $=[-74,761]$) and $\widehat\beta_1 = -3.4$ (95\% CI $=[-7.5,0.7]$). This indicates that activity suggestions at the beginning of the study may have a sizeable effect (increasing Week 6 daily step count by 344 steps), but this effect likely deteriorates over time (due to the negative $\widehat\beta_1$). This is intriguing: early interventions, though more distant from Week 6, appear to have larger effects on habit formation than those delivered later.

To assess the linearity assumption [where we set $f(t,s) = (1,t-1)^T$], we conducted another moderation analysis with $f(t,s)$ containing an intercept and a B-spline basis (6 degrees of freedom), allowing $\tau(t,s)$ to vary flexibly with $t$. For interpretability, we constructed the B-spline basis on $\text{day}_t := \lfloor \frac{t-1}{5}\rfloor + 1$. \cref{fig:heartsteps-analysis} shows an increasing trend in the DCEE during the first week, peaking at the end of Week 1, then gradually decreasing and crossing 0 at around the end of Week 2. 

\begin{figure}[htbp]
    \centering
    \includegraphics[width = 0.7\textwidth]{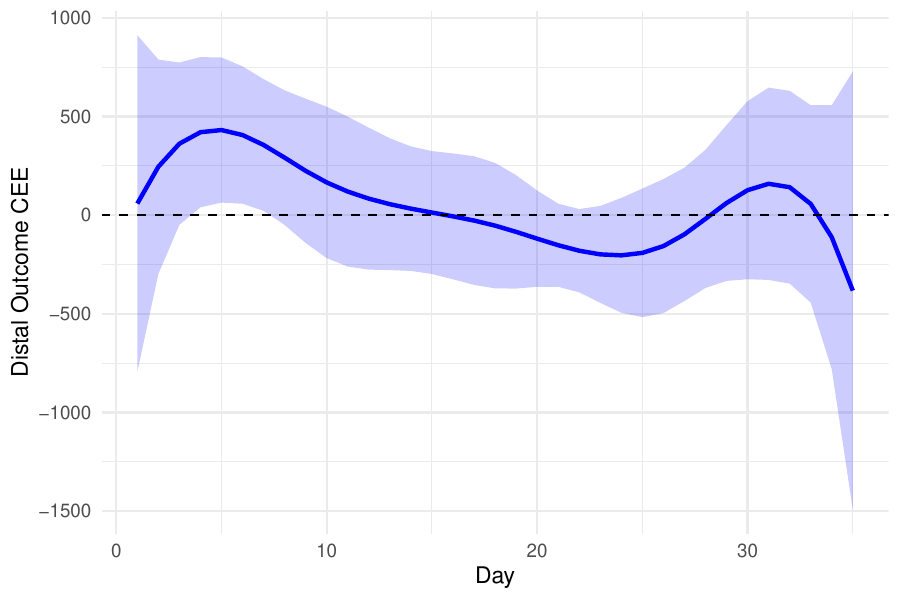}
    \caption{Time-varying effect of a push notification on the distal outcome (Week 6 average daily steps) across study days in the HeartSteps MRT. Solid curve: point estimate; shaded area: 95\% pointwise confidence intervals.}
    \label{fig:heartsteps-analysis}
\end{figure}

This finding suggests that early activity suggestions in the HeartSteps study may play a crucial role in shaping long-term behavior change, even if they are more distant from the distal outcome. One possible explanation is that early interventions help establish a momentum for behavior change, laying a foundation for participants to internalize and sustain new activity patterns. These early nudges may have a lasting impact by fostering psychological mechanisms such as self-efficacy, intrinsic motivation, or goal-setting.

In contrast, activity suggestions provided later in the study show diminishing returns, for two possible reasons that align with habit formation theory \citep{lally2013promoting}. First, later in the study some participants may have already developed a habit and thus rely less on prompts. Second, some participants may have disengaged due to message fatigue or perceived redundancy in the prompts. Our finding suggests that future mHealth programs should prioritize timely support early on to promote long-term habit formation.


Lastly, we note that the activity suggestions were also randomly delivered in Week 6, and thus the distal outcome was also influenced by those interventions. Our analysis remains valid despite of this dependence, but the results should be interpreted in the context of the Week 6 interventions. In other words, if micro-randomizations had stopped at the end of Week 5, the DCEE in that hypothetical scenario could be different due to potential interactions between Week 6 interventions and earlier interventions.

\section{Discussion}
\label{sec:discussion}

We introduced the distal causal excursion effect (DCEE), a novel causal estimand for evaluating the long-term impact of time-varying treatments in MRTs. We developed two estimators, both robust to outcome model misspecification. When applied to HeartSteps MRT, our found that early-in-the-study activity suggestions may have a stronger long-term impact on habit formation. Overall, this work provides the much-needed toolkit for scientists developing digital interventions to assess long-term causal effects using MRT data.

\revision{The causal excursion effect (CEE) and DCEE offer valuable but distinct insights for optimizing future digital interventions. CEE provides information on short-term responsiveness, while DCEE quantifies long-term impacts, and together they can inform outcome selection, intervention design, and iterative refinement. Moreover, they can help assess whether proximal outcomes serve as meaningful surrogates for distal goals and guide how short-term and long-term objectives might be balanced in both offline and real-time adaptive settings. We provide a detailed discussion of these practical and scientific considerations in Supplementary Material~\ref{A-sec:CEE_vs_DCEE}.}

Future work could extend our approach in several directions. First, our method assumes no unmeasured confounding and a correctly specified propensity score, and both assumptions are satisfied in MRTs but not necessarily in observational studies; incorporating doubly robust estimation and methods to address unmeasured confounding would improve applicability to observational studies. Second, the current linear model for the DCEE could be extended to data-adaptive models, especially if the set of moderators $S_t$ is large. Third, our approach estimates an average causal effect across all participants; future research could adapt it to estimate person-specific heterogeneous effects. Fourth, linking proximal and distal outcomes through mediation analysis could help clarify mechanisms driving long-term behavior change. Fifth, extending the DCEE to allow general reference policies and multiple-step excursions would further broaden the applicability of the method.

\revision{Finally, future work could explore policy learning approaches to identify optimal treatment sequences that maximize long-term benefit. However, the numerous decision points pose challenges, including an exponentially growing policy space, possible positivity violations, and instability of inverse probability weights. Solutions may build upon ideas from reinforcement learning and dynamic treatment regime literature, such as making structural assumptions (e.g., Markov or restricted history dependence), searching within parameterized policy classes (e.g., using policy gradient methods), and compressing history into lower-dimensional representations using domain knowledge.}

\section*{Acknowledgement}
\label{sec:acknowledgement}

The author thanks the anonymous referees, the Associate Editor, and the Editor for their conscientious efforts and constructive comments which improved the quality of this paper.

\section*{Supplementary Materials}
\label{sec:supp}

Supplementary Material \ref{A-sec:dcee-examples} contains the extensions and details for the DCEE examples in \cref{subsec:def-cee}. Supplementary Material \ref{A-sec:proof-identification} proves the identification result \cref{thm:identification}. Supplementary Material \ref{A-sec:beta-star-projection} proves the weighted average form of $\beta^\star$ given in \cref{rmk:beta-as-weighted-average}. Supplementary Material \ref{A-sec:est-fcn-post-proj-derivation} derives the estimating function $\phi(\beta,\mu)$ in \cref{eq:est-fcn-post-proj}. Supplementary Material \ref{A-sec:proof-normality} proves the asymptotic normality \cref{thm:normality}. Supplementary Material \ref{A-sec:simulation-true-parameter} details the numerical computation of true parameter values in simulation studies. Supplementary Material \ref{A-sec:CEE_vs_DCEE} discusses practical considerations for using CEE and DCEE to inform intervention design.

\section*{Data Availability}

The HeartSteps I MRT data that support the findings of this paper are publicly available at \url{https://github.com/klasnja/HeartStepsV1}. The code to reproduce the simulation and data analysis results can are available at \url{https://github.com/tqian/paper_DCEE/tree/v1.0}. The proposed estimator for distal causal excursion effect (DCEE) is implemented in the \texttt{dcee()} function in R package \texttt{MRTAnalysis} \citep{RpackageMRTAnalysis}.

\bibliographystyle{agsm}

\bibliography{main-and-appendix-distal-outcome-MRT}

\newpage


\begin{appendices}

\spacingset{1.9} 

\section{Extensions and Details for the DCEE Examples in \texorpdfstring{\cref{subsec:def-cee}}{Section 2.2}}
\label{A-sec:dcee-examples}

We provide the detailed derivation for the four DCEE Examples in \cref{subsec:def-cee}, where because \cref{asu:causal-assumptions} holds in all these examples, we can invoke \cref{thm:identification} and in particular \cref{eq:identify-cee-ie2} to calculate $\tau(t)$. We also present more general versions of \cref{ex:direct-indirect-effects} and \cref{ex:eligibility} to illustrate additional nuanced effects that are captured by the DCEE.

\subsection{\texorpdfstring{\cref{ex:marginalize-over-effect-modifiers}}{Example 1}: DCEE marginalizes over effect modifiers not in \texorpdfstring{$S_t$}{St}}

Consider the data-generating process described in \cref{ex:marginalize-over-effect-modifiers}. Below we show that $\tau(t) = \alpha_t + \beta_t ~\EE(X_t)$.

Given the exogeneity of $X_t$ and $A_t$, and that $I_t \equiv 1$, we have
\begin{align*}
    \EE(Y \mid H_t, I_t = 1, A_t = 1) - \EE(Y \mid H_t, I_t = 1, A_t = 0) = \alpha_t + \beta_t X_t,
\end{align*}
and the form of $\tau(t)$ follows immediately from \cref{eq:identify-cee-ie2}.

\subsection{\texorpdfstring{\cref{ex:treatment-burden}}{Example 2}: DCEE captures user burden}

Consider the data-generating process described in \cref{ex:treatment-burden}. Below we show that $\tau(s) = \beta_s - (\alpha_{s-1} + \alpha_s) p$ with $\alpha_0 := 0$, and the form of $\frac{1}{T}\sum_{t=1}^T \tau(t)$ in \cref{ex:treatment-burden} follows immediately.

Given the exogeneity of $X_t$ and $A_t$, and that $I_t \equiv 1$, we have
\begin{align*}
    & ~~~~ \EE(Y \mid H_s, A_s) \\
    & = \EE\{g(\bX_T) \mid H_s, A_s\} + \sum_{t=1}^T \beta_t \EE(A_t \mid H_s, A_s) - \sum_{t=1}^{T-1} \alpha_t \EE(A_t A_{t+1} \mid H_s, A_s) \\
    & = \EE\{g(\bX_T) \mid H_s\} + \sum_{t=1}^{s-1} \beta_t A_t + \beta_s A_s + \sum_{t=s+1}^{T} \beta_t \EE(A_t \mid H_s) \\
    & ~~~~ - \sum_{t=1}^{s-2} \alpha_t A_t A_{t+1} - \alpha_{s-1} A_{s-1} A_s - \alpha_s A_s \EE(A_{s+1} \mid H_s) - \sum_{t = s+1}^T \alpha_t \EE(A_t A_{t+1} \mid H_s).
\end{align*}
Therefore,
\begin{align*}
    \EE(Y \mid H_s, A_s = 1) - \EE(Y \mid H_s, A_s = 0) = \beta_s - \alpha_{s-1} A_{s-1} - \alpha_s p.
\end{align*}
Plugging this into \cref{eq:identify-cee-ie2} and we have
\begin{align*}
    \tau(s) = \EE(\beta_s - \alpha_{s-1} A_{s-1} - \alpha_s p) = \beta_s - (\alpha_{s-1} + \alpha_s) p.
\end{align*}
Averaging over $s \in [T]$ and we have the form of $\frac{1}{T}\sum_{t=1}^T \tau(t)$ in \cref{ex:treatment-burden}.

\subsection{\texorpdfstring{\cref{ex:direct-indirect-effects}}{Example 3} (extension): DCEE combines direct and indirect effects}

We present a more general version of \cref{ex:direct-indirect-effects}. Consider a two-timepoint setting ($T=2$), where $X_1$ is exogeneous, all individuals are always eligible ($I_1 = I_2 \equiv 1$), treatments $A_1, A_2 \sim \bern(p)$ are exogeneous, and the treatment $A_1$ influences the time-varying covariate $X_2$ through $X_2 \mid A_1 \sim \bern(\rho_0 + \rho_1 A_1)$. Generalizing \cref{ex:direct-indirect-effects}, we assume that the outcome is generated from $Y = \gamma_0 + \gamma_1 X_1 + \gamma_2 X_2 + \beta_1 A_1 + \beta_2 A_2 + \alpha_1 X_1 A_1 + \alpha_2 X_2 A_2 + \epsilon$.

We derive $\tau(1)$ below. We have
\begin{align*}
    \EE(Y \mid X_1, A_1) & = \gamma_0 +\gamma_1 X_1 + \gamma_2 \EE(X_2 \mid A_1) + \beta_1 A_1 + \beta_2 \EE(A_2) + \alpha_1 X_1 A_1 + \alpha_2 \EE(X_2 \mid A_1) \EE(A_2) \\
    & = \gamma_0 +\gamma_1 X_1 + (\gamma_2 + \alpha_2 p)(\rho_0 + \rho_1 A_1) + \beta_1 A_1 + \beta_2 p + \alpha_1 X_1 A_1.
\end{align*}
Therefore, 
\begin{align*}
    \EE(Y \mid X_1, A_1 = 1) - \EE(Y \mid X_1, A_1 = 0) = (\gamma_2 + \alpha_2 p) \rho_1 + \beta_1 + \alpha_1 X_1.
\end{align*}
Plugging this into \cref{eq:identify-cee-ie2} and we have
\begin{align*}
    \tau(1) = \beta_1 + \alpha_1 \EE(X_1) + (\gamma_2 + \alpha_2 p) \rho_1.
\end{align*}

We interpret each term in $\tau(1)$:
\begin{itemize}
    \item $\beta_1$: the direct of $A_1$ on $Y$ not moderated by $X_1$;
    \item $\alpha_1 \EE(X_1)$: the direct of $A_1$ on $Y$ moderated by $X_1$;
    \item $\gamma_2 \rho_1$: the indirect of $A_1$ on $Y$ mediated by $X_2$ that does not interact with $A_2$;
    \item $\alpha_2 p \rho_1$: the indirect of $A_1$ on $Y$ mediated by $X_2$ that interacts with $A_2$.
\end{itemize}

\subsection{\texorpdfstring{\cref{ex:eligibility}}{Example 4} (extension): DCEE captures treatment effects on future eligibility}

We present a more general version of \cref{ex:eligibility}. Consider a two-timepoint setting ($T=2$), where all individuals are eligible at $t=1$ ($I_1 \equiv 1$), treatment at $t=1$ influences future eligibility: $I_2 \mid A_1 \sim \bern(\rho_0 - \rho_1 A_1)$, and the treatment assignments follow $A_1 \sim \bern(p)$ and $A_2 \mid I_2 = 1 \sim \bern(p)$. For simplicity assume that there are no other covariates besides $I_t$. Generalizing \cref{ex:eligibility}, we assume that the outcome is generated from $Y = \beta_0 + \beta_1 A_1 + \beta_2 A_2 - \alpha A_1 A_2 + \epsilon$. Assume $\beta_1, \beta_2, \alpha, \rho_0, \rho_1 > 0$.

We derive $\tau(1)$ below. We have
\begin{align*}
    \EE(Y \mid A_1) & = \beta_0 + \beta_1 A_t + (\beta_2 - \alpha A_1) \EE(A_2 \mid A_1).
\end{align*}
For $\EE(A_2 \mid A_1)$ we have
\begin{align*}
    \EE(A_2 \mid A_1) = \EE \{ \EE(A_2 \mid I_2, A_1) \mid A_1\} = \EE(I_2 p \mid A_1) = p(\rho_0 - \rho_1 A_1).
\end{align*}
Plugging into $\EE(Y \mid A_1)$ and we have
\begin{align*}
    \EE(Y \mid A_1) = \beta_0 + \beta_1 A_t + (\beta_2 - \alpha A_1) p(\rho_0 - \rho_1 A_1).
\end{align*}
Therefore, 
\begin{align*}
    & ~~~~ \EE(Y \mid A_1 = 1) - \EE(Y \mid A_1 = 0) \\ 
    & = \beta_1 + (\beta_2 - \alpha) p(\rho_0 - \rho_1) - \beta_2 p\rho_0 \\
    & = \beta_1 - p \rho_0 \alpha + p \rho_1 \alpha - p \rho_1 \beta_2.
\end{align*}
Plugging this into \cref{eq:identify-cee-ie2} and we have
\begin{align*}
    \tau(1) = \beta_1 - p \rho_0 \alpha + p \rho_1 \alpha - p \rho_1 \beta_2.
\end{align*}

We interpret each term in $\tau(1)$:
\begin{itemize}
    \item $\beta_1$: the direct of $A_1$ on $Y$;
    \item $-p \rho_0 \alpha$: the negative interaction between $A_1$ and $A_2$ not accounting for the influence of $A_1$ on $I_2$;
    \item $p \rho_1 \alpha$: the reduced negative interaction of $A_1 \times A_2$ due to $A_1$ negatively impacting future eligibility $I_2$;
    \item $- p \rho_1 \beta_2$: the reduced effectiveness of $A_2$ on $Y$ due to $A_1$ negatively impacting future eligiblity $I_2$.
\end{itemize}

\section{Proof of Identification Result (\texorpdfstring{\cref{thm:identification}}{Theorem 1})}
\label{A-sec:proof-identification}

We first state and prove a useful lemma.

\begin{lem}
    \label{A-lem:identification-proofuse}
    We have
    \begin{itemize}
        \item[(a)] $p_t^a\{d_t^a(H_t) \mid H_t \} = 1$.
        \item[(b)] $p_t^a\{1 - d_t^a(H_t) \mid H_t \} = 0$.
    \end{itemize}
\end{lem}

\begin{proof}[Proof of \cref{A-lem:identification-proofuse}]
    When $a=1$, by definition we have $p_t^1(1\mid H_t) = I_t = \indic(I_t = 1)$ and $p_t^1(0\mid H_t) = 1 - I_t = \indic(I_t = 0)$. This implies that $p_t^1(a\mid H_t) = \indic(I_t = a)$. We also have $d_t^1(H_t) = I_t$. Therefore, $p_t^1\{d_t^1(H_t) \mid H_t\} = p_t^1(I_t \mid H_t) = 1$.

    When $a = 0$, by definition we have $d_t^0(H_t) = 0$ and $p_t^0\{d_t^0(H_t) \mid H_t\} = p_t^0(0 \mid H_t) = 0$. This proves statement (a).

    Statement (b) follows immediately from statement (a) and the fact that
    \begin{align*}
        p_t^a\{d_t^a(H_t) \mid H_t\} + p_t^a\{1 - d_t^a(H_t) \mid H_t\} = 1.
    \end{align*}

    This completes the proof.
\end{proof}

\revision{We now state and prove an equivalent version of \cref{thm:identification}, using slightly different notation to facilitate the proof. Let $p^1_t(a \mid H_t)$ denote the probability mass function of $A_t$ conditional on $H_t$ under the decision rule $d_t^1$: $p^1_t(a \mid H_t) = P_{A_t \sim d_t^1(H_t)}(A_t = a \mid H_t)$, i.e., $p^1_t(1 \mid H_t) = I_t$ and $p^1_t(0 \mid H_t) = 1 - I_t$. Similarly, let $p^0_t(a \mid H_t)$ denote the distribution of $A_t$ under $d_t^0$, i.e., $p^0_t(1 \mid H_t) = 0$ and $p^0_t(0 \mid H_t) = 1$. The following \cref{A-thm:identification} is equivalent to \cref{thm:identification}: one can directly verify that $\dfrac{p_t^1(A_t \mid H_t)}{p_t(A_t \mid H_t)} = \dfrac{\indic(A_t = I_t)}{P(A_t = I_t \mid H_t)}$ and $\dfrac{p_t^0(A_t \mid H_t)}{p_t(A_t \mid H_t)} = \dfrac{\indic(A_t = 0)}{P(A_t = 0 \mid H_t)}$ by checking equivalence under each scenario of $(A_t,I_t)\in\{0,1\}^{\otimes 2}$.}

\begin{thm}
    \label{A-thm:identification}
    Under \cref{asu:causal-assumptions}, we have
    \begin{align}
        \tau(t,s) 
        & = \EE \bigg\{ \frac{p_t^1(A_t \mid H_t)}{p_t(A_t \mid H_t)} Y - \frac{p_t^0(A_t \mid H_t)}{p_t(A_t \mid H_t)} Y ~\bigg|~ S_t = s \bigg\}. \label{A-eq:identify-cee-ipw} \\
        & = \EE \Big[ \EE\big\{Y \mid H_t, A_t = d_t^1(H_t)\big\} - \EE\big\{Y \mid H_t, A_t = d_t^0(H_t)\big\} \mid S_t = s\Big] \label{A-eq:identify-cee-ie} \\
        & = \EE(I_t) \EE \Big\{ \EE(Y \mid H_t, I_t = 1, A_t = 1) - \EE(Y \mid H_t, I_t = 1, A_t = 0) \mid S_t = s, I_t = 1 \Big\}. \label{A-eq:identify-cee-ie2}        
    \end{align}
\end{thm}

\begin{proof}[Proof of \cref{A-thm:identification}]
    Repeatedly using the law of iterated expectations, we have
    \begin{align}
        & ~~~ \EE \bigg\{ \frac{p_t^a(A_t \mid H_t)}{p_t(A_t \mid H_t)} Y ~\bigg|~ S_t \bigg\} \label{A-eq:identification-proofuse0} \\
        & = \EE \bigg[ \EE \bigg\{ \frac{p_t^a(A_t \mid H_t)}{p_t(A_t \mid H_t)} Y ~\bigg|~ H_t \bigg\} ~\bigg|~ S_t \bigg] \nonumber \\
        & = \EE \bigg[ \EE \bigg\{ \frac{\overbrace{p_t^a(A_t \mid H_t)}^{= 1 \text{ due to \cref{A-lem:identification-proofuse}(a)}}}{p_t(A_t \mid H_t)} Y ~\bigg|~ H_t, A_t = d_t^a(H_t) \bigg\} ~p_t\{d_t^a(H_t) \mid H_t \} ~\bigg|~ S_t \bigg] \nonumber \\
        & ~~~ + \EE \bigg[ \EE \bigg\{ \frac{\overbrace{p_t^a(A_t \mid H_t)}^{= 0 \text{ due to \cref{A-lem:identification-proofuse}(b)}}}{p_t(A_t \mid H_t)} Y ~\bigg|~ H_t, A_t = 1 - d_t^a(H_t) \bigg\} ~p_t\{1 - d_t^a(H_t) \mid H_t \} ~\bigg|~ S_t \bigg] \nonumber \\
        & = \EE [ \EE\{ Y \mid H_t, A_t = d_t^a(H_t)\} \mid S_t ] \label{A-eq:identification-proofuse1} \\
        & = \EE ( \EE [ \EE\{ Y \mid H_t, A_t = d_t^a(H_t), \bA_{t+1:T}^{a_t = d_t^a(H_t)}\} \mid H_t, A_t = d_t^a(H_t) ] \mid S_t ) \nonumber \\
        & = \EE ( \EE [ \EE\{ Y (\bd_{t-1}, d_t^a, \bd_{t+1:T}) \mid H_t, A_t = d_t^a(H_t), \bA_{t+1:T}^{a_t = d_t^a(H_t)}\} \mid H_t, A_t = d_t^a(H_t) ] \mid S_t )\label{A-eq:identification-proofuse2} \\
        & = \EE ( \EE [ \EE\{ Y (\bd_{t-1}, d_t^a, \bd_{t+1:T}) \mid H_t, A_t = d_t^a(H_t)\} \mid H_t, A_t = d_t^a(H_t) ] \mid S_t ) \label{A-eq:identification-proofuse3} \\
        & = \EE [ \EE\{ Y (\bd_{t-1}, d_t^a, \bd_{t+1:T}) \mid H_t, A_t = d_t^a(H_t)\} \mid S_t ] \nonumber \\
        & = \EE [ \EE\{ Y (\bd_{t-1}, d_t^a, \bd_{t+1:T}) \mid H_t\} \mid S_t ] \label{A-eq:identification-proofuse4} \\
        & = \EE\{ Y (\bd_{t-1}, d_t^a, \bd_{t+1:T}) \mid S_t(\bd_{t-1}) \}. \label{A-eq:identification-proofuse5}
    \end{align}
    Here, \cref{A-eq:identification-proofuse1} follows from \cref{A-lem:identification-proofuse}, \cref{A-eq:identification-proofuse2} follows from \cref{asu:consistency}, \cref{A-eq:identification-proofuse3} follows from \cref{asu:sequential-ignorability}, \cref{A-eq:identification-proofuse4} follows from \cref{asu:sequential-ignorability}, \cref{A-eq:identification-proofuse5} follows from \cref{asu:consistency}.

    The fact that \cref{A-eq:identification-proofuse0,A-eq:identification-proofuse1,A-eq:identification-proofuse5} are equal to each other establishes \cref{A-eq:identify-cee-ipw,A-eq:identify-cee-ie}.

    Finally, to prove \cref{A-eq:identify-cee-ie2}, we have 
    \begin{align}
        & ~~~ \EE [ \EE\{ Y \mid H_t, A_t = d_t^1(H_t)\} - \EE\{ Y \mid H_t, A_t = d_t^0(H_t)\} \mid S_t ] \nonumber \\
        & = \EE [ \EE\{ Y \mid H_t, A_t = d_t^1(H_t)\} - \EE\{ Y \mid H_t, A_t = d_t^0(H_t)\} \mid S_t, I_t = 1 ] P(I_t = 1) \nonumber \\
        & ~~~ + \EE [ \EE\{ Y \mid H_t, A_t = d_t^1(H_t)\} - \EE\{ Y \mid H_t, A_t = d_t^0(H_t)\} \mid S_t, I_t = 0 ] P(I_t = 0) \nonumber \\
        & = \EE \{ \EE( Y \mid H_t, I_t = 1, A_t = 1) - \EE( Y \mid H_t, I_t = 1, A_t = 0) \mid S_t, I_t = 1 \} P(I_t = 1) + 0, \nonumber
    \end{align}
    and this establishes \cref{A-eq:identify-cee-ie2}.
            
    The proof is thus completed.
\end{proof}

\section{The Weighted Average Form of \texorpdfstring{$\beta^\star$}{beta star}}
\label{A-sec:beta-star-projection}

We state and prove a general theorem for the weighted average form of $\beta^\star$ defined in \cref{eq:beta-star-def}, which immediately implies \cref{eq:beta-star-projection-1,eq:beta-star-projection-2}.

\begin{thm}
    \label{A-thm:beta-star-projection-general}
    Under \cref{asu:causal-assumptions}, for the $\beta^\star$ defined in \cref{eq:beta-star-def} we have
    \revision{\begin{align}
        \beta^\star = & \bigg[\sum_{t=1}^T \omega(t) \EE \{f(t,S_t)f(t,S_t)^T\} \bigg]^{-1} \nonumber\\
        & \times \sum_{t=1}^T \omega(t) \EE(I_t) \EE \Big[ \big\{ \EE(Y | H_t, I_t = 1, A_t = 1) - \EE(Y | H_t, I_t = 1, A_t = 0) \mid S_t, I_t = 1 \big\} f(t,S_t) \Big]. \label{A-eq:beta-star-projection-general}
    \end{align}}
\end{thm}

\begin{proof}[Proof of \cref{A-thm:beta-star-projection-general}]
    Let $L(\beta)$ be the loss function in \cref{eq:beta-star-def}, i.e.,
    \begin{align*}
        L(\beta) := \sum_{t=1}^T \omega(t) \EE\Big[\{ \tau(t,S_t) - f(t,S_t)^T \beta \}^2\Big].
    \end{align*}
    The right hand side of \cref{A-eq:beta-star-projection-general} is the $\beta^\star$ that minimizes the quadratic form $L(\beta)$. We prove this now.

    We have
    \begin{align*}
        L(\beta) = \sum_{t=1}^T \omega(t) \EE \{ \tau(t,S_t)^t - 2\tau(t,S_t)f(t,S_t)^T\beta + \beta^T f(t,S_t)f(t,S_t)^T \beta \}.
    \end{align*}
    Taking the derivative with respect to $\beta$ and set to 0 at $\beta = \beta^*$, we have
    \begin{align}
        \frac{\partial L(\beta^\star)}{\partial\beta^T} = \sum_{t=1}^T \omega(t) \EE \{ - 2\tau(t,S_t)f(t,S_t)^T + 2 (\beta^\star)^T f(t,S_t)f(t,S_t)^T \} = 0. \label{A-eq:beta-star-score-equation}
    \end{align}
    This implies that
    \begin{align}
        \beta^\star = \bigg[\sum_{t=1}^T \omega(t) \EE \{f(t,S_t)f(t,S_t)^T\} \bigg]^{-1} \bigg[\sum_{t=1}^T \omega(t) \EE \{\tau(t,S_t)f(t,S_t)\} \bigg]. \label{A-eq:beta-star-projection-general-proofuse1}
    \end{align}
    Plugging \cref{eq:identify-cee-ie2} from \cref{thm:identification} into \cref{A-eq:beta-star-projection-general-proofuse1} and we immediately get \cref{A-eq:beta-star-projection-general}. This completes the proof.
\end{proof}

\section{Derivation of the Estimating Function \texorpdfstring{$\phi(\beta,\mu)$}{phi} in \texorpdfstring{\cref{eq:est-fcn-post-proj}}{Eq (10)}}
\label{A-sec:est-fcn-post-proj-derivation}

We restate necessary definition. \revision{We rewrite \cref{eq:est-fcn-pre-proj} using the equivalence between \cref{thm:identification} and \cref{A-thm:identification}, to facilitate the derivation:}
\begin{align}
    \xi(\beta) = \sum_{t=1}^T \omega(t) \bigg\{ \frac{p_t^1(A_t \mid H_t)}{p_t(A_t \mid H_t)} Y - \frac{p_t^0(A_t \mid H_t)}{p_t(A_t \mid H_t)} Y - f(t, S_t)^T\beta \bigg\} f(t,S_t). \label{A-eq:est-fcn-pre-proj}
\end{align}
Motivated by \citet{robins1999testing}, we subtract from $\xi(\beta)$ its projection on the score functions of the treatment selection probabilities to obtain a more efficient estimating function:
\begin{align}
    \xi(\beta) - \sum_{u=1}^T \Big[ \EE\{\xi(\beta) \mid H_u, A_u\} - \EE\{\xi(\beta) \mid H_u\} \Big]. \label{A-eq:est-fcn-post-proj-derivation-xi}
\end{align}
Define $\xi_t(\beta)$ to be a summand in \cref{A-eq:est-fcn-pre-proj} so that $\xi(\beta) = \sum_{t=1}^T \xi_t(\beta)$:
\begin{align}
    \xi_t(\beta) := \omega(t) \bigg\{ \frac{p_t^1(A_t \mid H_t) - p_t^0(A_t \mid H_t)}{p_t(A_t \mid H_t)} Y - f(t, S_t)^T\beta \bigg\} f(t,S_t). \label{A-eq:est-fcn-post-proj-derivation-xit}
\end{align}
Then \cref{A-eq:est-fcn-post-proj-derivation-xi} becomes
\begin{align}
    & ~~~ \sum_{t=1}^T \bigg( \xi_t(\beta) - \sum_{u=1}^T \Big[ \EE\{\xi_t(\beta) \mid H_u, A_u\} - \EE\{\xi_t(\beta) \mid H_u\} \Big] \bigg) \nonumber \\
    & = \sum_{t=1}^T \bigg( \xi_t(\beta) - \Big[ \EE\{\xi_t(\beta) \mid H_t, A_t\} - \EE\{\xi_t(\beta) \mid H_t\} \Big] \bigg) \label{A-eq:est-fcn-post-proj-derivation-proofuse1} \\
    & ~~~ - \sum_{t=1}^T \bigg( \sum_{1 \leq u \leq T, u \neq t} \Big[ \EE\{\xi_t(\beta) \mid H_u, A_u\} - \EE\{\xi_t(\beta) \mid H_u\} \Big] \bigg). \label{A-eq:est-fcn-post-proj-derivation-proofuse2}
\end{align}

Let $\phi_t(\beta,\mu_t)$ be a summand in \cref{eq:est-fcn-post-proj} so that $\phi(\beta,\mu) = \sum_{t=1}^T \phi_t(\beta, \mu_t)$:
\begin{align}
    \phi_t(\beta, \mu_t) := \omega(t) \bigg[ I_t \frac{(-1)^{1-A_t}}{p_t(A_t \mid H_t)} \Big\{ Y - p_t(0\mid H_t) \mu_t(H_t, 1) - p_t(1\mid H_t) \mu_t(H_t, 0) \Big\} - f(t, S_t)^T\beta \bigg] f(t,S_t). \label{A-eq:est-fcn-post-proj-derivation-phit}
\end{align}
It follows from \cref{A-lem:est-fcn-post-proj-derivation-intermediate-step}, which we will establish below, that \cref{A-eq:est-fcn-post-proj-derivation-proofuse1} equals $\phi(\beta,\mu) = \sum_{t=1}^T \phi_t(\beta, \mu_t)$, the proposed estimating function with improved efficiency. The terms in \cref{A-eq:est-fcn-post-proj-derivation-proofuse2} cannot be analytically derived without imposing additional models on the relationship between current and lagged variables in the longitudinal trajectory, thus we omit them when deriving the improved estimating function. This same heuristic was employed in \citet{cheng2023efficient} and \citet{bao2024estimating}. Therefore, it suffices to establish the following lemmas.

\begin{lem}
    \label{A-lem:est-fcn-post-proj-derivation-proofuse}
    We have 
    \begin{itemize}
        \item[(a)] $p_t^1(A_t \mid H_t) = A_t I_t + (1-I_t)$. 
        \item[(b)] $p_t^0(A_t \mid H_t) = 1 - A_t$.
        \item[(c)] $p_t^1(A_t \mid H_t) - p_t^0(A_t \mid H_t) = A_t I_t + A_t - I_t$.
        \item[(d)] $\dfrac{p_t^1(A_t \mid H_t) - p_t^0(A_t \mid H_t)}{p_t(A_t \mid H_t)} = I_t\dfrac{(-1)^{1-A_t}}{p_t(A_t \mid H_t)}$.
    \end{itemize}
\end{lem}

\begin{proof}[Proof of \cref{A-lem:est-fcn-post-proj-derivation-proofuse}]
    When $I_t = 1$, $p_t^1(1 \mid H_t) = 1$ and $p_t^1(0 \mid H_t) = 0$. When $I_t = 0$, $p_t^1(1 \mid H_t) = 0$ and $p_t^1(0 \mid H_t) = 1$. Therefore, $p_t^1(A_t \mid H_t) = A_t I_t + (1-I_t)$. This proves (a).

    For $p_t^0(a \mid H_t)$, regardless of $I_t$, we have $p_t^0(1 \mid H_t) = 0$ and $p_t^0(0 \mid H_t) = 1$. Therefore, $p_t^0(A_t \mid H_t) = 1 - A_t$. This proves (b).

    (c) follows immediately from (a) and (b).

    To prove (d), first note that $p_t(A_t \mid H_t) = A_t p_t(1 \mid H_t) + (1 - A_t) \{1 - p_t(1 \mid H_t)\}$. This combined with (b) gives
    \begin{align*}
        \frac{p_t^1(A_t \mid H_t) - p_t^0(A_t \mid H_t)}{p_t(A_t \mid H_t)} = \frac{A_t I_t + A_t - I_t}{A_t p_t(1 \mid H_t) + (1 - A_t) \{1 - p_t(1 \mid H_t)\}},
    \end{align*}
    which takes values $\frac{1}{p_t(1 \mid H_t)}$, $- \frac{1}{1 - p_t(1 \mid H_t)}$, and 0 when $I_t = A_t = 1$, $I_t = 1$ and $A_t = 0$, and $I_t = 0$ (in which case $A_t = 0$), respectively. We can directly verify that $I_t\dfrac{(-1)^{1-A_t}}{p_t(A_t \mid H_t)}$ also takes these exact same values under these three scenarios, respectively. This proves (d).

    This completes the proof.
\end{proof}

\begin{lem}
    \label{A-lem:est-fcn-post-proj-derivation-intermediate-step}
    Fix $t \in [T]$. For $\xi_t(\beta)$ defined in \cref{A-eq:est-fcn-post-proj-derivation-xit} and $\phi_t(\beta, \mu_t)$ defined in \cref{A-eq:est-fcn-post-proj-derivation-phit}, we have
    \begin{align}
        \phi_t(\beta, \mu_t^\star) = \xi_t(\beta) - \EE\{ \xi_t(\beta) \mid H_t, A_t\} + \EE\{ \xi_t(\beta) \mid H_t\}. \label{A-eq:est-fcn-post-proj-derivation-intermediate-step}
    \end{align}
\end{lem}

\begin{proof}[Proof of \cref{A-lem:est-fcn-post-proj-derivation-intermediate-step}]
    By definition of $\xi_t(\beta)$ and \cref{A-lem:est-fcn-post-proj-derivation-proofuse}(d), we have
    \begin{align}
        \xi_t(\beta) := \omega(t) \bigg\{ I_t\dfrac{(-1)^{1-A_t}}{p_t(A_t \mid H_t)} Y - f(t, S_t)^T\beta \bigg\} f(t,S_t). \label{A-eq:est-fcn-post-proj-derivation-xit-2}
    \end{align}
    Therefore,
    \begin{align}
        \EE\{\xi_t(\beta) \mid H_t, A_t\} = \omega(t) \bigg\{ I_t\dfrac{(-1)^{1-A_t}}{p_t(A_t \mid H_t)} \EE(Y \mid H_t, A_t) - f(t, S_t)^T\beta \bigg\} f(t,S_t). \label{A-eq:est-fcn-post-proj-derivation-proofuse0}
    \end{align}
    We also have
    \begin{align}
        \EE\{\xi_t(\beta) \mid H_t\} & = \EE\{\xi_t(\beta) \mid H_t, I_t = 1\} I_t + \EE\{\xi_t(\beta) \mid H_t, I_t = 0\} (1-I_t) \nonumber \\
        & = \EE\{\xi_t(\beta) \mid H_t, I_t = 1, A_t = 1\} P(A_t = 1 \mid H_t, I_t = 1) I_t \nonumber \\
        & ~~~ + \EE\{\xi_t(\beta) \mid H_t, I_t = 1, A_t = 0\} P(A_t = 0 \mid H_t, I_t = 1) I_t \nonumber \\
        & ~~~ + \EE\{\xi_t(\beta) \mid H_t, I_t = 0, A_t = 0\} (1-I_t) \nonumber \\
        & = \omega(t) \bigg\{ I_t \frac{1}{p_t(1 \mid H_t, I_t = 1)} ~ \EE(Y \mid H_t, I_t = 1, A_t = 1) ~ p_t(1 \mid H_t, I_t = 1) \nonumber \\
        & ~~~ + I_t \frac{-1}{p_t(0 \mid H_t, I_t = 1)} ~ \EE(Y \mid H_t, I_t = 1, A_t = 0) ~ p_t(0 \mid H_t, I_t = 1) \nonumber \\
        & ~~~ + 0 - f(t, S_t)^T\beta \bigg\} f(t,S_t) \label{A-eq:est-fcn-post-proj-derivation-proofuse3} \\
        & = \omega(t) \Big[ I_t \{ \EE(Y \mid H_t, I_t = 1, A_t = 1) - \EE(Y \mid H_t, I_t = 1, A_t = 0) \} - f(t, S_t)^T\beta \Big] f(t,S_t), \label{A-eq:est-fcn-post-proj-derivation-proofuse4}
    \end{align}
    where \cref{A-eq:est-fcn-post-proj-derivation-proofuse3} follows from the equivalent form of $\xi(\beta)$ in \cref{A-eq:est-fcn-post-proj-derivation-xit-2}.

    Putting together \cref{A-eq:est-fcn-post-proj-derivation-xit-2,A-eq:est-fcn-post-proj-derivation-proofuse0,A-eq:est-fcn-post-proj-derivation-proofuse4}, we have
    \begin{align}
        & ~~~ \xi_t(\beta) - \EE\{ \xi_t(\beta) \mid H_t, A_t\} + \EE\{ \xi_t(\beta) \mid H_t\} \nonumber \\
        & = \omega(t) \bigg( I_t \bigg[\frac{(-1)^{1-A_t}}{p_t(A_t \mid H_t)} \{Y - \EE(Y \mid H_t, A_t)\} + \EE(Y \mid H_t, I_t = 1, A_t = 1) \nonumber \\
        & ~~~ - \EE(Y \mid H_t, I_t = 1, A_t = 0) \} \bigg] - f(t, S_t)^T\beta \bigg) f(t, S_t) \nonumber \\
        & = \omega(t) \bigg( I_t \bigg[\frac{(-1)^{1-A_t}}{p_t(A_t \mid H_t)} \{Y - \mu_t^\star(H_t, A_t)\} + \mu_t^\star(H_t, 1) - \mu_t^\star(H_t, 0) \bigg] - f(t, S_t)^T\beta \bigg) f(t, S_t) \label{A-eq:est-fcn-post-proj-derivation-proofuse5} \\
        & = \omega(t) \bigg[ I_t \frac{(-1)^{1-A_t}}{p_t(A_t \mid H_t)} \Big\{ Y - p_t(0\mid H_t) \mu_t^\star(H_t, 1) - p_t(1\mid H_t) \mu_t^\star(H_t, 0) \Big\} - f(t, S_t)^T\beta \bigg] f(t,S_t), \label{A-eq:est-fcn-post-proj-derivation-proofuse6}
    \end{align}
    where \cref{A-eq:est-fcn-post-proj-derivation-proofuse5} follows from the definition of $\mu_t^\star(H_t, A_t)$, and \cref{A-eq:est-fcn-post-proj-derivation-proofuse6} follows from directly verifying the equality when $A_t = 1$ and when $A_t = 0$. \cref{A-eq:est-fcn-post-proj-derivation-proofuse6} is $\phi_t(\beta, \mu_t^\star)$ defined in \cref{A-eq:est-fcn-post-proj-derivation-phit}. This completes the proof.
\end{proof}

\section{Proof of Asymptotic Normality (\texorpdfstring{\cref{thm:normality}}{Theorem 2})}
\label{A-sec:proof-normality}


We first state and proof a useful lemma.

\begin{lem}
    \label{A-lem:normality-proofuse}
    We have
    \begin{align*}
        \EE\bigg\{ \frac{(-1)^{1-A_t}}{p_t(A_t \mid H_t)} ~\bigg|~ H_t, I_t = 1 \bigg\} = 0.
    \end{align*}
\end{lem}

\begin{proof}[Proof of \cref{A-lem:normality-proofuse}]
    By direct calculation we have
    \begin{align*}
        & ~~~\EE\bigg\{ \frac{(-1)^{1-A_t}}{p_t(A_t \mid H_t)} ~\bigg|~ H_t, I_t = 1 \bigg\} \\
        & = \EE\bigg\{ \frac{(-1)^{1-A_t}}{p_t(A_t \mid H_t)} ~\bigg|~ H_t, I_t = 1, A_t = 1 \bigg\} p_t(1 \mid H_t) + \EE\bigg\{ \frac{(-1)^{1-A_t}}{p_t(A_t \mid H_t)} ~\bigg|~ H_t, I_t = 1, A_t = 0 \bigg\} p_t(0 \mid H_t) \\
        & = \frac{1}{p_t(1 \mid H_t)} p_t(1 \mid H_t) + \frac{-1}{p_t(0 \mid H_t)} p_t(0 \mid H_t) = 0.
    \end{align*}
    This completes the proof.
\end{proof}

Now we prove \cref{thm:normality}.

\begin{proof}[Proof of \cref{thm:normality}]
    We show that the estimating function $\phi(\beta, \mu)$ is globally robust, i.e., $\EE\{\phi(\beta^\star, \mu)\} = 0$ for all $\mu$. Then, the asymptotic normality of $\widehat\beta$ and $\widetilde\beta$ follows immediately from Theorems 5.1 and 5.2 of \citet{cheng2023efficient}.

    By the definition of $\phi$ in \cref{eq:est-fcn-post-proj}, we have
    \begin{align}
        & ~~~\EE\{\phi(\beta^\star, \mu)\} \nonumber \\
        & = \sum_{t=1}^T \omega(t) \EE \bigg[\bigg\{ I_t \frac{(-1)^{1-A_t}}{p_t(A_t \mid H_t)} Y - f(t, S_t)^T\beta^\star \bigg\} f(t,S_t)\bigg] \nonumber \\
        & - \sum_{t=1}^T \omega(t) \EE\bigg(\bigg[ I_t \frac{(-1)^{1-A_t}}{p_t(A_t \mid H_t)} \Big\{ p_t(0\mid H_t) \mu_t(H_t, 1) + p_t(1\mid H_t) \mu_t(H_t, 0) \Big\} \bigg] f(t,S_t)\bigg). \label{A-eq:normality-proofuse1}
    \end{align}
    The second term in \cref{A-eq:normality-proofuse1} is 0 for any $\mu_t$ due to \cref{A-lem:normality-proofuse} and the law of iterated expectations. For the first term in \cref{A-eq:normality-proofuse1}, we have
    \begin{align}
        & ~~~ \EE \bigg\{ I_t \frac{(-1)^{1-A_t}}{p_t(A_t \mid H_t)} Y f(t,S_t)\bigg\} \nonumber \\
        & = \EE \bigg\{ \frac{p_t^1(A_t \mid H_t) - p_t^0(A_t \mid H_t)}{p_t(A_t \mid H_t)} Y f(t,S_t)\bigg\} \label{A-eq:normality-proofuse2} \\
        & = \EE\bigg[ \EE \bigg\{ \frac{p_t^1(A_t \mid H_t) - p_t^0(A_t \mid H_t)}{p_t(A_t \mid H_t)} Y ~\bigg|~ S_t\bigg\} f(t,S_t)  \bigg] \nonumber \\
        & = \EE \{ \tau(t,S_t) f(t,S_t) \}, \label{A-eq:normality-proofuse3}
    \end{align}
    where \cref{A-eq:normality-proofuse2} follows from \cref{A-lem:est-fcn-post-proj-derivation-proofuse}(d) and \cref{A-eq:normality-proofuse3} follows from \cref{thm:identification}. Plugging \cref{A-eq:normality-proofuse3} to the first term in \cref{A-eq:normality-proofuse1}, we have
    \begin{align}
        & ~~~ \sum_{t=1}^T \omega(t) \EE \bigg[\bigg\{ I_t \frac{(-1)^{1-A_t}}{p_t(A_t \mid H_t)} Y - f(t, S_t)^T\beta^\star \bigg\} f(t,S_t)\bigg] \nonumber \\
        & = \sum_{t=1}^T \omega(t) \EE \{\tau(t, S_t)f(t,S_t) - f(t,S_t)^T\beta^\star f(t,S_t)\} = 0, \nonumber
    \end{align}
    where the last equality is \cref{A-eq:beta-star-score-equation}. Therefore, we have proved that $\EE\{\phi(\beta^\star, \mu)\} = 0$ for all $\mu$. We now invoke Theorems 5.1 and 5.2 of \citet{cheng2023efficient} to immediately get \cref{thm:normality}.
\end{proof}

\section{Numerical Computation of True Parameter Values in Simulation Studies}
\label{A-sec:simulation-true-parameter}

The computation uses \cref{A-eq:beta-star-projection-general-proofuse1}. In particular, for each $t \in [T]$, we generated two data sets under two excursion policies, each of sample size 1 million: one excursion policy is $D_{d_t = d_t^1}$ and the other is $D_{d_t = d_t^0}$. $\tau(t, S_t = \emptyset)$ is computed as the difference in means of $Y$ from the two data sets. $\tau(t, Z_t = 0)$ and $\tau(t, Z_t = 1)$ are computed as the difference in  means of $Y$ from the two data sets among individuals with $Z_t = 0$ and with $Z_t = 1$, respectively. Then, a third data set of sample size 1 million is generated under the MRT policy, and the expectations in \cref{A-eq:beta-star-projection-general-proofuse1} are computed separately for the two sets of estimands using this third data set. This yields the numeric values of $\beta^\star_1$ and $(\beta^\star_2, \beta^\star_3)$.

\section{Practical Considerations for Using CEE and DCEE to Inform Intervention Design}
\label{A-sec:CEE_vs_DCEE}

From an applied perspective, both the causal excursion effect (CEE) and the distal causal excursion effect (DCEE) offer valuable but distinct information for optimizing future digital interventions. We summarize our preliminary view of their roles and how they can be used together to inform design and evaluation below. We thank an anonymous reviewer whose question motivated this section.

\subsection{Batch (offline) optimization perspective}

\textbf{Informing warm-start policies.} Insights from CEE and DCEE could inform ``warm-start'' policies for future trials or deployments, potentially balancing short-term gains with long-term sustainability. For example, in a physical activity study, distal outcomes (e.g., long-term activity level) may not be the only goal; proximal outcomes (e.g., frequent activity bouts) can also improve participants' attitudes towards activity or mental health in general. In such settings, a single scalar distal outcome may not be the only objective, and it becomes even more meaningful to consider both the distal outcome and the proximal outcomes. Designing such composite warm-start policies remains an open question, with some emerging literature in the reinforcement learning community \citep{shalev2016long,yang2024learning,wu2024policy}.

\textbf{Validating proximal outcomes as surrogates.} CEE focuses on short-term (proximal) responses, while DCEE captures long-term (distal) impacts. Together, they can help evaluate whether a chosen proximal outcome serves as a meaningful surrogate for the ultimate distal goal. If a strong proximal effect (CEE) does not correspond to a strong distal effect (DCEE), this suggests misalignment and motivates revision of the proximal outcome definition. Further methodology development may borrow ideas from the surrogate endpoint literature, such as focusing on proportion of treatment explained \citep{freedman1992statistical}, mediation-based approaches \citep{elliott2023surrogate}, or principal stratification-based approaches \citep{li2010bayesian}. One challenge is the presence of time-varying treatments and the fact that the candidate ``surrogacy'' is a time-varying proximal outcome in MRT.

\subsection{Online (real-time) optimization perspective}

\textbf{Reward shaping and learning speed.} In real-time adaptive settings (e.g., reinforcement learning-based intervention adaptation and delivery), focusing on CEE or carefully designed proximal rewards can accelerate learning. This acts as a form of reward shaping \citep{ng1999policy,hu2020learning}, shifting the optimization focus from distal feedback to immediate signals, when the distal feedback is delayed.

\textbf{Combining proximal and distal information.} An open problem is how to best integrate CEE and DCEE in online learning settings, especially when distal outcomes are only partially observed during early stages (e.g., some participants enrolled early have distal outcomes observed, while others enrolled later only proximal data).

\subsection{Feasibility considerations and iterative design approach}

In early-phase or pilot trials, distal outcomes may be infeasible to collect due to long follow-up times. In such cases, CEE becomes the only feasible estimand, and primary analysis must rely solely on proximal outcomes.

An iterative design approach may involve first optimizing short-term responsiveness using CEE, then validating and refining interventions using DCEE to ensure long-term impact. Ideally, intervention design, proximal outcomes, and distal outcomes should be aligned; discrepancies among them may indicate the need to revise one or more elements. Precisely characterizing such discrepancies and determining approach revisions remain important open questions.

\subsection{Insights for scientific understanding}

Beyond optimizing interventions, CEE and DCEE help advance behavioral science and validate theoretical mechanisms. CEE alone indicates whether the immediate impact is as intended. DCEE alone indicates whether there is a desired long-term impact. Together, they provide evidence toward understanding potential mediation pathways, with CEE representing immediate behavioral change and DCEE capturing sustained effects. While formal mediation analysis is required for definitive causal pathway evaluation (e.g., as we preliminarily explored in \citet{qian2025dynamic}), jointly examining CEE and DCEE offers valuable insights.

\subsection{Statistical perspective and bias-variance trade-offs}

CEE and DCEE can both be viewed as parsimonious parameterizations of a fully conditional long-term effect (e.g., modeling $\EE\{Y(\ba_T)\}$ or even $\EE\{Y(\ba_T) \mid H_T(\ba_{T-1})\}$ for all possible $\ba_T$). Proximal outcomes could be easier to influence and may carry stronger signals (larger CEE), while distal outcomes could be harder to shift and may be noisier (smaller DCEE). However, one may argue that DCEE might be mroe similar to the fully conditional long-term effect given that both focus on the same distal outcome. This reflects a bias-variance trade-off spectrum:

\begin{center}
\begin{tikzpicture}[
    node distance=1.8cm,
    every node/.style={align=center},
    myarrow/.style={-{Stealth[length=5pt]}, thick}
]

\node (cee) {CEE};
\node (dcee) [right=of cee] {DCEE};
\node (full) [right=of dcee] {Full conditional\\long-term effect};

\draw[-] (cee) -- (dcee);
\draw[-] (dcee) -- (full);

\node[above=0.5cm of dcee, text width=5cm] (bias) {Bias decreases $\rightarrow$};

\node[below=0.5cm of dcee, text width=5cm] (var) {Variance increases $\rightarrow$};

\node[left=0.3cm of cee] (leftlabel) {(Smallest variance, \\ largest bias)};
\node[right=0.3cm of full] (rightlabel) {(Largest variance, \\ smallest bias)};
\end{tikzpicture}
\end{center}

\end{appendices}

\end{document}